\documentclass[submission,copyright,creativecommons]{eptcs}
\providecommand{\event}{EXPRESS/SOS 2019} % Name of the event you are submitting to
\usepackage{breakurl}             % Not needed if you use pdflatex only.
\usepackage{underscore}           % Only needed if you use pdflatex.

\usepackage{macros}

\title{Immediate Observation in Mediated Population Protocols}
\author{Tobias Prehn
\institute{Modelle und Theorie Verteilter Systeme\\
TU Berlin, Germany}
\email{tobias.prehn@tu-berlin.de}
\and
Myron Rotter
\institute{{}\\{TU Berlin, Germany}}
\email{m.rotter@campus.tu-berlin.de}
}

\def\titlerunning{Immediate Observation in Mediated Population Protocols}
\def\authorrunning{T. Prehn \& M. Rotter}
\begin{document}
\maketitle

\begin{abstract}
In this paper we analyze the computational power of variants of population protocols (PP), a formalism for distributed systems with anonymous agents having very limited capabilities.
The capabilities of agents are enhanced in mediated population protocols (MPP) by recording the states in the edges of the interaction graph.
Restricting the interactions to the communication model of immediate observation (IO) reduces the computational power of the resulting formalism.
%It is well-known that population protocols are able to compute exactly semilinear predicates.
%This power can be enhanced to symmetric predicates in $NSPACE(n^2)$ by allowing states in the edges of the interaction graph as done by mediated population protocols (MPP).
%If restricted to the communication model of immediate observation (IO), the basic population protocol capabilities are reduced to \COUNT{*}, \ie, boolean combinations of simple threshold predicates.
We show that this enhancement and restriction, when combined, yield a model (IOMPP) at least as powerful as the basic PP.
The proof requires a novel notion of configurations in the MPP model allowing differentiation of agents and uses techniques similar to methods of analyzing encoding criteria, namely operational correspondence.
The constructional part of the proof is generic in a way that all protocols can be translated into the new model without losing the desirable properties they might have besides a stable output.
Furthermore, we illustrate how this approach could be utilized to prove our conjecture of IOMPP model being even as expressive as the MPP model.
If our conjecture holds, this would result in a sharp characterization of the computational power and reveal the nonnecessity of two-way communication in the context of mediated population protocols.
\end{abstract}

\section{Introduction}
    
Population protocols have been introduced in 2004 as a computational model for passively mobile finite state sensors by Angluin \etal \cite{AAD04, AAD06}.
They feature a finite state space, making them suitable for computation units with very limited capabilities and full anonymity, resulting directly from this restriction.
Since the number of possible states that each agent could be in may not grow with the number of participating agents, there is no space for memorizing the ids of already met communication partners or similar constructs.
Therefore, the outcome of any binary communication does not depend on whether the participants have communicated before.
Another feature is the fully distributed approach of the base version for population protocols that does not need a base station, leader, or scheduler of any kind.
The impact of such extensions has been studied \cite{AAF08, DLF17, AGV15}.

It is well known that predicates computable by population protocols are exactly the semilinear predicates.
The first study on the computational power of this model was in 2007 by Angluin \etal \cite{AAE07}.
In this context, also several different communication patterns have been modeled in population protocols and their computational power have been studied as well.
One of those mechanisms has been the immediate observation model, which is a special kind of one-way communication as opposed to the two-way communication that comes with the base model.
The idea is that an agent may observe another agent without it noticing being observed.
Clearly the observed agent cannot change its state in such an interaction whereas the observer can use the information given by its own and the observed agent's state.
In contrast to stronger mechanisms no synchronization between the communication partners is needed.
Consequently the communication in such a model is asynchronous and applicable to a broader variety of systems.
With the fully distributed setting in mind the immediate observation communication seems to be a desirable feature.
But these qualities come with a price.
Protocols with this limitation to the communication can only compute predicates in \COUNT{*}, \ie predicates that count multiplicities of input values and compare them to previously given thresholds.

Another approach to altering population protocols has been the work of Michail \etal \cite{MCS11}.
In their model agents are allowed to store distinct information for different communication partners.
To achieve this they extended the base formalism by states for each pair of agents residing in the edges of the interaction graph.
Some previous work on mediated population protocols modelled directed interaction graphs with one state per edge \cite{MCS11} and others used undirected graphs where each edge has a state for each of its two endpoints \cite{DDF17}.
This extension is a reasonable compromise between maintaining the anonymity of each agent and being able to memorize the already met communication partners.
An agent is capable of telling an agent, that it has not yet communicated with, apart from one it has already met.
But two other agents being in the same state and with the same communication history are still indistinguishable.
Aside from this the edge states can be used for storing several other information.
This mediated population protocols are able to compute all symmetric predicates in $NSPACE(n^2)$.

We now present a model in this paper that combines the extended storage possibilities of mediated population protocols and the limited communication model of immediate observation protocols.
The computational power of our resulting formalism has to be studied as it is unclear how this extension and restriction interact.

In section~\ref{sec:tech-pre} the basic formalisms and existing models are defined.
We are using a representation of population that allows the distinction of agents from a global point of view.
This does not interfere with the anonymity of the agents and is for analysis purposes only.
Based on this we define our model of immediate observation mediated population protocols (IOMPP) in section~\ref{sec:iompp}.
Subsequently we study the computational power of our model in section~\ref{sec:exp-power}.
We take the approach simulating population protocols in immediate observation mediated population protocols in \ref{subsec:simulation}.
Additionally, we give a translation of configurations from one model to the other and define criteria such a translation has to meet for it to express desirable attributes in \ref{subsec:transl-and-criteria}.
Our work is inspired by and makes use of the encodability criteria stated by Gorla in~\cite{G10}.
To the best of our knowledge this technique is novel to population protocols in the way we utilize it in \ref{subsec:iompp-exp-power} to prove that immediate observation mediated population protocols are at least as expressive as the base model of population protocols.
In \ref{subsec:IOMPP-vs-MPP} we conjecture that our approach could also be used to show that immediate observation does not restrict the computational power of mediated population protocols.
We conclude our paper by a discussion on the given results and possible application of the used techniques in section~\ref{sec:conclusion-and-futer-work}.
We also give an outlook on future work and open questions.
\section{Technical preliminaries}\label{sec:tech-pre}
First we introduce populations, which form the base of all population protocols.
They are often modeled as multisets to emphasize the indistinguishability of the participating agents.
We will use vectors where each entry represents the state of a specific agent, because we want to efficiently compare populations in different models from a global viewpoint.
Note that this will not give agents a distinct id they could make use from their local point of view.
A different kind of vector representation can be found in \cite{EGL17} and is not to be confused with ours.
They use vectors where each entry describes for an agent state the multiplicities of agents in that state.
Their vector representation efficiently stores sets of agents with their states, making it easy to identify equivalent protocol states.
In contrast to this, our representation can be used to compare sets of agents in population protocols with sets of agents in extended variants like mediated population protocols.
\\
\begin{definition}[Populations]
Let $A$ be a nonempty finite set and $n \in \NAT$.
Then $A^n$ describes the set of $n$-tuples over $A$ also referred to as vectors of length $n$.
A population over $A$, denoted by $\POP{A}$, is the set of all vectors of arbitrary but finite length over $A$.
If $v \in A^n$ we use $\LENGTH{v} = n$ to denote the length of $v$ and $\ELEM{v}{i}$ to reference the $i^{th}$ element of $v$ (with $i \in \NAT^{+}_{\leq n}$).
\end{definition}
Next we define how calculations in population protocols are modeled from a local point of view.
Each agent has the same set of states and rules how to change states according to a communication partner's state.
Additionally, functions to map input to an initial state and a state to some output are defined to be identical for each agent.
\begin{definition}[Population Protocols \cite{AR09}]
A population protocol $P$ is a 5-tuple $P = (Q, \Sigma, \LOCIN, \LOCOUT, \LOCTRANS)$ where $Q$ is a finite set of agent states, $\Sigma$ a finite input alphabet, $\LOCIN: \Sigma \to Q$ describes the input function, $\LOCOUT: Q \to \{0, 1\}$ is the output function, and $\LOCTRANS: Q^2 \to Q^2$ is referred to as transition function and describes all possible pairwise interactions.
We also making use of a set representation of $\LOCTRANS \subseteq Q^4$ whenever we write $t = (p, q) \to_\LOCTRANS (p', q')$ referring to a specific transition $t \in \LOCTRANS$.
\end{definition}
To analyze protocols we need a global perspective.
Here states are configurations holding state information for each agent.
This kind of global view has been used to study the computational power of population protocols \cite{AAE07}.
\begin{definition}[Global Protocols \cite{AAE07}]
Let $P = (Q, \Sigma, \LOCIN, \LOCOUT, \LOCTRANS)$ be a population protocol.
The global protocol to $P$ is a 5-tuple $G_P = (\CONFIGS, \Sigma, \GLOBIN, \GLOBOUT, \GLOBTRANS)$ where $\CONFIGS = \POP{Q}$  is the set of configurations, \ie, vectors of agent states $Q$, $\GLOBIN: \POP{\Sigma} \to \CONFIGS$ maps input vectors to initial configurations, $\GLOBOUT: \CONFIGS \to \{0, 1, \bot \}$ maps configurations to outputs, and $\GLOBTRANS: \CONFIGS \to \CONFIGS$ is the global transition function with $\GLOBTRANS^{*}$ being its reflexive and transitive closure.
For $C, C' \in \CONFIGS$ it holds that $C \GLOBTRANS C'$ iff there is a transition $t \in \LOCTRANS$ and $i, j \in \NAT^{+}$ with $i \not= j$ such that $t = (\ELEM{C}{i}, \ELEM{C}{j}) \to_\LOCTRANS (\ELEM{C'}{i}, \ELEM{C'}{j})$ and $\ELEM{C}{k} = \ELEM{C'}{k}$ for every $k \in \NAT^{+}\setminus\{i, j\}$.
We also write $C \GLOBTRANSLAB{t_{i,j}} C'$ and call agent $i$ the \emph{initiator} of $t$ and $j$ the \emph{responder} or (if the state of $i$ is not changed by $t$) \emph{observer}.
The global input function takes use of $\LOCIN$ to get an agent state for each single value in its input and $\GLOBOUT$ aggregates the outputs of the agents according to $\LOCOUT$.
It holds that $\GLOBOUT(C) = x \in \{0, 1\}$ iff $\LOCOUT(\ELEM{C}{i}) = x$ for each $i \in \NAT^{+}_{\leq \LENGTH{C}}$ and $\GLOBOUT(C) = \bot$ in every other case.
When the underlying protocol $P$ is clear from the context we often omit the index of $G_P$ and simply state that $G$ is the global protocol to $P$.
\end{definition}

Based on a global protocol we can describe what it means for a protocol to compute some predicate.
For this we need to define executions and fairness.
\begin{definition}[Computation]
Let $G = (\CONFIGS, \Sigma, \GLOBIN, \GLOBOUT, \GLOBTRANS)$ be a global protocol.
A configuration $C \in \CONFIGS$ is output stable with output $x \in \{0, 1\}$ iff $\GLOBOUT(C') = x$ for each $C' \in \CONFIGS$ with $C \GLOBTRANS^{*} C'$.
We call a sequence of configurations $C_0, C_1, C_2, \dots \in \CONFIGS$ with $C_i \GLOBTRANS C_{i+1}$ for each $i \in \NAT$ an execution.
An execution is fair iff for each $C \in \CONFIGS$ with $C_i = C$ for infinitely many $i \in \NAT$ it holds that if there is a transition $C \GLOBTRANS C'$ then also $C_j = C'$ for infinitely many $j \in \NAT$.
A population protocol $P$ is well-specified if for each input $\mathit{Inp}$, it holds that all fair executions of $P$ starting in $\GLOBIN\left(\mathit{Inp}\right)$ reach a configuration that is output stable.
$P$ computes a predicate if this reached configuration is output stable with output $1$ if $\mathit{Inp}$ satisfies the predicate and with output $0$ otherwise.
\end{definition}
In the context of population protocols several communication mechanisms have been studied \cite{AAE07}.
Immediate observation is one of those mechanisms.
It reduces the class of computable predicates to predicates counting multiplicities of input values \COUNT{*}.
 To model this kind of communication, restrictions to the allowed form of transitions are made.
\begin{definition}[Immediate Observation]
Let $P = (Q, \Sigma, \LOCIN, \LOCOUT, \LOCTRANS)$ be a population protocol.
$P$ is an immediate observation protocol, if there is no transition that changes the state of the initiator.
In other words, all transitions $t \in \LOCTRANS$ have to be of the form $t = (p, q) \to_\LOCTRANS (p, q')$.
\end{definition}
Another extension to population protocols is the mediated variant.
The idea is to introduce states in all edges of the communication graph.
Since the most general graph is the complete graph, each pair of agents is given such a state.
In the context of an immediate observation communication mechanism, it is not reasonable to assume a storage that both agents can write to.
Therefore, we introduce a pair of edge states for each pair of agents.
Edge states are always initialized with the same value.
\begin{definition}[Mediated Population Protocols \cite{MCS11}]
A mediated population Protocol $P$ is a 7-tuple $P = (Q, \Sigma, S, s_0, \LOCIN, \LOCOUT, \LOCTRANS)$ where $Q, \Sigma, \LOCIN$ and $\LOCOUT$ are analogous to population protocols. The set of edge states $S$ includes the initial edge state $s_0 \in S$ and the transition function $\LOCTRANS: \left(Q \times S\right)^2 \to \left(Q \times S\right)^2$ incorporates the edge states for each pair of agents.
\end{definition}
Configurations in mediated population protocols cannot be represented by simple vectors.
We need to introduce matrices as configurations containing the agent states on the diagonal and the states of the edge between agents $a$ and $b$ in fields $C_{a,b}$ (side of agent $a$) and $C_{b,a}$ (side of agent $b$).
\begin{definition}[Mediated Populations]
$A^{n \times n}$ describes the set of square matrices over $A$ of size $n \times n$.
A mediated population over $A$ denoted by $\POPM{A}$ is the set of all matrices of arbitrary but finite length over $A$.
If $m \in A^{n \times n}$ we use $\ELEM{m}{i,j}$ to reference the element of $m$ at column $i$ and row $j$ with $i,j \in \NAT^{+}_{\leq n}$ and $\LENGTH{m} = n$ to denote the length as well as the height of a square matrix $m$.
\end{definition}
We can now proceed with lifting our global protocol definitions to represent mediated population protocols as well.
\begin{definition}[Global Protocols for Mediated Population Protocols]
Let $P = (Q, \Sigma, S, s_0, \LOCIN, \LOCOUT, \LOCTRANS)$ be a mediated population protocol.
The global protocol to $P$ is again a 5-tuple $G = (\CONFIGS, \Sigma, \GLOBIN, \GLOBOUT, \GLOBTRANS)$.
In contrast to global protocols for simple population protocols $\CONFIGS = \POPM{Q}$ is the set of configurations and $\GLOBIN: \POP{\Sigma} \to \CONFIGS$ maps input vectors to initial configurations, initializing the diagonal fields with the corresponding agent states and every other field with $s_0$.
The output function $\GLOBOUT: \CONFIGS \to \{0, 1, \bot \}$ ignores all fields not on the diagonal and $\GLOBTRANS: \CONFIGS \to \CONFIGS$ now also changes the respective edge states.
For $C, C' \in \CONFIGS$ it holds that $C \GLOBTRANS C'$ iff there is a transition $t \in \LOCTRANS$ and $i, j \in \NAT^{+}$ with $i \not= j$ such that $t = (\ELEM{C}{i,i}, \ELEM{C}{i,j}, \ELEM{C}{j,j}, \ELEM{C}{j,i}) \to_\LOCTRANS (\ELEM{C'}{i,i}, \ELEM{C'}{i,j}, \ELEM{C'}{j,j}, \ELEM{C'}{j,i})$ and $\ELEM{C'}{k,l} = \ELEM{C'}{k,l}$ for every $k,l \in \NAT^{+}\setminus\{i, j\}$.
\end{definition}
\section{Modelling immediate observation in mediated population protocols}\label{sec:iompp}
From the technical preliminaries in section~\ref{sec:tech-pre} we can easily combine the models for mediated population protocols and immediate observation conform communication.
We get our model of population protocols with two edge states in every edge, one per communication partner, and transitions that keeps the states of the initiator unaltered and changes the states of the observer.
\begin{definition}[Immediate Observation Mediated Population Protocols]
An immediate observation mediated population protocol $P$ is a 7-tuple $P = (Q, \Sigma, S, s_0, \LOCIN, \LOCOUT, \LOCTRANS)$ where $Q$ is a finite set of agent states, $\Sigma$ a finite input alphabet, $\LOCIN: \Sigma \to Q$ describes the input function, $\LOCOUT: Q \to \{0, 1\}$ is the output function, and $\LOCTRANS: \left(Q \times S\right)^2 \to \left(Q \times S\right)^2$ is referred to as transition function and describes all possible pairwise interactions.
We also making use of a set representation of $\LOCTRANS \subseteq \left(Q \times S\right)^4$ whenever we write $t = (p, s, q, r) \to_\LOCTRANS (p', s', q', r')$ referring to a specific transition $t \in \LOCTRANS$.
Since our model uses the immediate observation communication mechanism, all transitions $t = (p, s, q, r) \to_\LOCTRANS (p', s', q', r')$ have to satisfy $p = p'$ and $s = s'$.
\end{definition}
\subsection{Simulating population protocols by immediate observation mediated population protocols}\label{subsec:simulation}
We can now simulate protocols in the basic population protocol model by immediate observation mediated population protocols.
The main idea is to split every two-way communication with an initiator and a responder into 4 steps.
Two steps are required to signal the request and the acknowledgement of a communication and another two steps are needed to finish the communication resolving all pending state changes.
Additionally, a reset transition is given for the case of an unsuccessful communication.
\begin{simulation}\label{sim:PP-to-IOMPP}
Let $P = (Q, \Sigma, \LOCIN, \LOCOUT, \LOCTRANS)$ be a population protocol. The following immediate observation mediated population protocol $P'$ simulates the protocol $P$ and is given by the tuple $(Q',\Sigma',S',s'_0,\LOCIN',\LOCOUT',\LOCTRANS')$ where
\begin{eqnarray*}
    \Sigma' &:=& \Sigma, \\
    Q' &:=& \{L,U\} \times Q, \\
    S' &:=& \{\SINIT, \SPONR\} \cup Q, \\
    s'_0 &:=& \SINIT, \\
    \LOCIN'(\sigma) &:=& (U,\LOCIN(\sigma)) \text{ for all } \sigma \in \Sigma, \\
    \LOCOUT'(l,q) &:=& \LOCOUT(q) \text{ for all } (l,q) \in Q'.
\end{eqnarray*}
The states of the agents are of the form $(l,q)$ where $l\in \{L,U\}$ indicates whether the agent is locked or unlocked and $q \in Q$ is the computation state according to the original population protocol $P$.
For easier referencing we call the first component of an agent state $(l, q)$ the locking state $l$ of this agent and the second component its computation state $q$.
We use the formulation of an agent being locked whenever its locking state is $L$ and say this agent is unlocked otherwise.
W.l.o.g. we assume that $\{\SINIT, \SPONR\} \cap Q = \emptyset$.
The input function $\LOCIN$ maps each input symbol $\sigma \in \Sigma$ to the state $(U,\LOCIN(\sigma))$.
The output function $\LOCOUT'$ maps each state $(l,q) \in Q$ to $\LOCOUT(q)$ independent of the locking indicator $l$.
We specify for each transition $t = (p,q) \rightarrow_{\LOCTRANS} (p',q')$ of $\LOCTRANS$ with $p,q,p',q' \in Q$ the following transitions for $\LOCTRANS'$.
\begin{eqnarray}
	t^{(1)} = ((U,p),\SINIT,(U,q),\SINIT) &\rightarrow_{\LOCTRANS'}& ((U,p),\SINIT,(L,q'),q) \\
	t^{(2)} = ((L,q'),q,(U,p),\SINIT) &\rightarrow_{\LOCTRANS'}& ((L,q'),q,(L,p'),\SPONR) \\
    t^{(3)} = ((L,p'),\SPONR,(L,q'),q) &\rightarrow_{\LOCTRANS'}& ((L,p'),\SPONR,(U,q'),\SINIT) \\
    t^{(4)} = ((x,y),\SINIT,(L,p'),\SPONR) &\rightarrow_{\LOCTRANS'}& ((x,y),\SINIT,(U,p'),\SINIT) \\
    &&\text{ for every } (x,y) \in Q' \nonumber \\
    t^{(5)} = ((x,y),z,(L,q'),q) &\rightarrow_{\LOCTRANS'}& ((x,y),z,(U,q),\SINIT) \\
    &&\text{ for every } (x,y) \in Q' \nonumber \\
    &&\text{ and } z \in S' \setminus \{\SPONR\} \nonumber
\end{eqnarray}
\end{simulation}
The locking state of each agent prohibits simultaneous participation in several different communications.
Whenever an agent took part in a two-way communication $t$ in the original protocol, it could be the observer of a transition of type (1) in the simulation.
A $t^{(1)}$ transition locks the observing agent and puts its old computation state in the edge state this agent controls on the edge with the observed agent.
This has two reasons: First it signals the interest in a communication with the other agent and second it backups the old state for a potential future reset.
If the other agent observes the change in the edge state, it may signal the acknowledgement of requested communication by locking itself, changing its computation state according to the transition $t$ and putting $\SPONR$ in its edge state.
This is achieved by transition $t^{(2)}$.
Now the two agents have to reset their edge states to $\SINIT$ and unlock themselves.
The agent mimicking the responder of the original transition $t$ starts by taking $t^{(3)}$, followed by the simulator of the original initiator taking $t^{(4)}$.
If a communication was not successful, either because the initiators surrogate has taken an other transition with another agent in the meantime or because responder simulating agent observes its partner before it could acknowledge the communication, $t^{(5)}$ is taken.
This transition assures that in the described cases an agent can give up on a communication attempt and reset its state, readying itself for another attempt, potentially with a different partner.

Note that if a transition $t = (p,q) \rightarrow_{\LOCTRANS} (p,q')$ is already immediate observation compliant we do not need to add the whole set of transitions.
We could instead add a slightly altered version of the original transition as follows.
\begin{eqnarray}
	t^{(6)} = ((U,p),\SINIT,(U,q),\SINIT) &\rightarrow_{\LOCTRANS'}& ((U,p),\SINIT,(U,q'),\SINIT)
\end{eqnarray}
Clearly the result would be the same.
This kind of transition is only needed if the simulation needs to be more efficient in the sense of steps needed to get to an output stable configuration.
We will therefore omit this type of transitions in our analyses.
\begin{observation}[Output Changing Transitions]\label{obs:output-changing-transitions}
By definition of $\LOCOUT'$ the output of an agent only depends on its computation state.
As transitions $t^{(3)}$, $t^{(4)}$ do not change the computation states, only transitions $t^{(1)}$, $t^{(2)}$, and $t^{(5)}$ can have an impact on the output of an agent.
Since Simulation~\ref{sim:PP-to-IOMPP} is an immediate observation protocol, this agent has to be the observer of such transitions.
\end{observation}
\begin{observation}[Number of Started Conversations]\label{obs:started-conversations}
Every agent has at most one started and not yet concluded conversation at any point in time.
Starting a conversation by taking transition $t^{(1)}$ as responder brings an agent to a locked state.
Therefore, no other conversation can be started or acknowledged by this agent until the conversation is concluded with transition $t^{(3)}$ or aborted with transition $t^{(5)}$.
Acknowledging a conversation by taking transition $t^{(2)}$ as responder also brings an agent to a locked state.
Again no other conversation can be started or acknowledged by this agent until the conversation is concluded with transition $t^{(4)}$ or aborted with transition $t^{(5)}$.
\end{observation}
\begin{observation}[Point of no Return]\label{obs:ponr}
Every occurrence of transition $t^{(2)}_{i,j}$ with acting agents $i$ and $j$ is eventually followed by transitions $t^{(3)}_{j,i}$ and $t^{(4)}_{i,j}$.
After execution of $t^{(2)}_{i,j}$ agent $i$ is locked with $q \in Q$ in its edge state to $j$ and agent $j$ is locked with $\SPONR$ in its edge state to $i$.
From Observation~\ref{obs:started-conversations} we know that neither $i$ nor $j$ can be observer of any transition with some agent different from $i$ and $j$.
From the transitions with $i$ and $j$ only $t^{(3)}_{j,i}$ is enabled and will be taken at some point because of the fairness assumptions in population protocols.
After that agent $j$ is still locked with $\SPONR$ in its edge state to $i$.
Therefore, $j$ can only be observer of transitions $t^{(3)}, t^{(4)}$, or $t^{(5)}$.
Again from Observation~\ref{obs:started-conversations} we know that only $t^{(4)}_{i,j}$ is possible.
\end{observation}
\section{Computational power}\label{sec:exp-power}
We now show that our model can compute all predicates computable in population protocols by giving a translation, that relates configurations from a protocol to configurations from its simulation representing the same state of computation.
Additionally, we identify requirements imposed on such a translation to be helpful in proving the equality of computed predicates.
We will ultimately show how this proof is executed.
\subsection{Translation and criteria}\label{subsec:transl-and-criteria}
We provide a translation, that constructs configurations in mediated population protocols from configurations of population protocols by giving all agents an unlocked state and setting all edges to the neutral $\SINIT$ state.
\begin{definition}[Translation of Configurations]\label{def:translation}
Let $P = (Q, \Sigma, \LOCIN, \LOCOUT, \LOCTRANS)$ be a population protocol and $P' = (Q',\Sigma',S',s'_0,\LOCIN',\LOCOUT',\LOCTRANS')$ a immediate observation mediated population protocol constructed from $P$ using Simulation~\ref{sim:PP-to-IOMPP}.
By $\TRANSL{\cdot}: \POP{Q} \to \POPM{Q'}$ we denote the translation of configurations $C \in \POP{Q}$ in the population protocol into configurations $D \in \POPM{Q'}$ from the mediated population protocol.
This translation is defined as follows:
\[\TRANSL{\left( q_0, q_1, \dots, q_n \right)} =
    \begin{pmatrix}
        (U,q_1) & \SINIT & \dots & \SINIT \\
        \SINIT & (U,q_2) & \ddots & \vdots \\
        \vdots & \ddots & \ddots & \SINIT \\
        \SINIT & \dots & \SINIT & (U,q_n)
    \end{pmatrix}
\]
\end{definition}
From the criteria for \emph{good encodings} defined by Gorla \cite{G10} we adopt the notion of operational correspondence.
\begin{definition}[Operational Correspondence]\label{def:operational-correspondence}
A translation $\TRANSL{\cdot}: \POP{Q} \to \POPM{Q'}$ is operationally corresponding if it is
\begin{enumerate}
    \item[(1)] (operationally) complete, \ie, for all $C, C' \in \POP{Q}$ with $C \GLOBTRANS^{*} C'$ it holds that $\TRANSL{C} \GLOBTRANS^{*} \TRANSL{C'}$, and
    \item[(2)] (operationally) sound, \ie, for all $C \in \POP{Q}$ and $D \in \POPM{Q'}$ with $\TRANSL{C} \GLOBTRANS^{*} D$ there exists a $C' \in \POP{Q}$ with $D \GLOBTRANS^{*} \TRANSL{C'}$ and $C \GLOBTRANS^{*} C'$.
\end{enumerate}
\end{definition}
If our translation in Def.~\ref{def:translation} instantiated with concrete population protocol $P$ and mediated population protocol $P'$ is operationally corresponding, we get that every configuration reachable in $P$ is also reachable in $P'$ and vice versa.
\begin{definition}[Input/Output Correspondence]
Let $P$ be a population protocol, $P'$ be a mediated population protocol and $G, G'$ be the global protocols to $P$ and $P'$ respectively. A translation $\TRANSL{\cdot}: \POP{Q} \to \POPM{Q'}$ is I/O corresponding if it is
\begin{enumerate}
    \item[(1)] input corresponding, \ie, for all $V \in \POP{\Sigma}$ it holds that $\TRANSL{\GLOBIN\left(V\right)} = \GLOBIN'\left(V\right)$, and
    \item[(2)] output corresponding, \ie, for all $C \in \POP{Q}$ it holds that $\GLOBOUT\left(C\right) = \GLOBOUT'\left(\TRANSL{C}\right)$.
\end{enumerate}
\end{definition}
Input/Output Correspondence gives us the assurance that input and output functions of the protocols related by a translation behave in a similar way.
If it holds, translating an input configuration of the original protocol or directly using the input function of the corresponding mediated population protocol yields the same result.
Additionally, configurations are always translated into configurations with the same output.
\begin{definition}[Output Stability Preservation]
Let $P$ be a population protocol, $P'$ be a mediated population protocol and $G, G'$ be the global protocols to $P$ and $P'$ respectively. A translation $\TRANSL{\cdot}: \POP{Q} \to \POPM{Q'}$ is output stability preserving if for each  $C \in \POP{Q}$ it holds that $\TRANSL{C}$ is output stable iff $C$ is output stable.
\end{definition}
From the output stability preservation we get that each output stable configuration is translated into a configuration also being output stable.
\begin{lemma}\label{lem:translation-gives-predicate-equality}
Let $P = (Q, \Sigma, \LOCIN, \LOCOUT, \LOCTRANS)$ be a population protocol and $P' = (Q',\Sigma',S',s'_0,\LOCIN',\LOCOUT',\LOCTRANS')$ a mediated population protocol.
If there is a translation $\TRANSL{\cdot}: \POP{Q} \to \POPM{Q'}$ that is operationally corresponding, input/output corresponding, and output stability preserving, then $P$ and $P'$ compute the same predicate.
\end{lemma}
\begin{proof}
Since the translation is input corresponding every initial configuration in $P$ has a corresponding initial configuration in $P'$.
From the operational correspondence we know that a configuration is reachable from an initial configuration in $P'$ iff it has a corresponding configuration reachable from the corresponding initial configuration in $P$.
With output correspondence both configurations clearly have the same output and, since the translation is output stability preserving, the output is either stable in both configurations or in none.
Therefore, the protocol $P$ calculates the same semilinear predicate as $P'$.
\end{proof}
\subsection{All semilinear predicates can be computed by immediate observation mediated population protocols}\label{subsec:iompp-exp-power}
\begin{lemma}\label{lem:operationally-corresponding-translation}
The translation $\TRANSL{\cdot}$ given in Def.~\ref{def:translation} is operationally corresponding for any population protocol $P = (Q, \Sigma, \LOCIN, \LOCOUT, \LOCTRANS)$ and the mediated population protocol $P' = (Q',\Sigma',S',s'_0,\LOCIN',\LOCOUT',\LOCTRANS')$ constructed using Simulation~\ref{sim:PP-to-IOMPP}.
\end{lemma}
\begin{proof}
To prove operational completeness assume that $C, C' \in \POP{Q}$ with $C \GLOBTRANS^* C'$.
Since $\GLOBTRANS^*$ is defined as reflexive-transitive closure of $\GLOBTRANS$ we get that $C_1, C_2, \dots, C_n \in \POP{Q}$ exist with $C \GLOBTRANS C_1 \GLOBTRANS C_2 \GLOBTRANS \dots \GLOBTRANS C_n \GLOBTRANS C'$.
We can always simulate a step $C_i \GLOBTRANS C_{i+1}$ in $P$ by making 4 steps in $P'$.
Assume that the step is due to transition $t \in \LOCTRANS$ and agents at $a$ and $b$ are acting as initiator and responder respectively in $C_i \GLOBTRANSLAB{t_{a,b}} C_{i+1}$.
This can be simulated as $\TRANSL{C_i} \GLOBTRANSLAB{t^{(1)}_{a,b}} C_i^2 \GLOBTRANSLAB{t^{(2)}_{b,a}} C_i^3 \GLOBTRANSLAB{t^{(3)}_{a,b}} C_i^4 \GLOBTRANSLAB{t^{(4)}_{b,a}} \TRANSL{C_{i+1}}$ and thus $\TRANSL{C_i} \GLOBTRANS^4 \TRANSL{C_{i+1}}$ holds.
Thus, $\TRANSL{C} \GLOBTRANS^4 \TRANSL{C_1} \GLOBTRANS^4 \dots \GLOBTRANS^4 \TRANSL{C_n} \GLOBTRANS^4 \TRANSL{C'}$ exemplifies $\TRANSL{C} \GLOBTRANS^* \TRANSL{C'}$.

To prove operational soundness assume that $C \in \POP{Q}$ and $D \in \POPM{Q'}$ with $\TRANSL{C} \GLOBTRANS^* D$.
We construct $C'$ as follows.
For this assume $(D)_{i,i} = (l_i, d_i)$.
\begin{equation*}
    (C')_i =
        \begin{cases}
            g & \text{, if there is exactly one } j \in \NAT^{+}_{\leq \LENGTH{D}, \not= i} \text{ such that } (D)_{i,j} = g \text{ and } (D)_{j,i} \not= \SPONR \\
            d_i & \text{, otherwise}
        \end{cases}
\end{equation*}
We can show that $D \GLOBTRANS^* \TRANSL{C'}$ by taking the appropriate transitions for all $i,j \in \NAT^{+}_{\leq \LENGTH{D}}$ with $i \not= j$ as follows.
{\def\arraystretch{1.5}
\begin{center}
\begin{tabular}{lcll}
    if $(D)_{i,j} = q$ & and & $(D)_{j,i} = \SPONR$ & take transitions $t^{(3)}_{j,i}, t^{(4)}_{i,j}$\\
    if $(D)_{i,j} = \SINIT$ & and & $(D)_{j,i} = \SPONR$ & take transition $t^{(4)}_{i,j}$\\
    if $(D)_{i,j} = q$ & and & $(D)_{j,i} \not= \SPONR$ & take transition $t^{(5)}_{j,i}$\\
    \multicolumn{3}{c}{otherwise} & do nothing
\end{tabular}
\end{center}
}
Since $\TRANSL{C} \GLOBTRANS^* D$ by assumption and $D \GLOBTRANS^* \TRANSL{C'}$ we get that $\TRANSL{C} \GLOBTRANS^* \TRANSL{C'}$.
We can construct $C \GLOBTRANS C_1 \GLOBTRANS C_2 \GLOBTRANS \dots \GLOBTRANS C'$ from the path $\TRANSL{C} \GLOBTRANS D_1 \GLOBTRANS D_2 \GLOBTRANS \dots  \GLOBTRANS \TRANSL{C'}$ as follows.
Whenever there is a step $D_x \GLOBTRANSLAB{t^{(2)}_{i,j}} D_y$ in this path, take transition $C_v \GLOBTRANSLAB{t_{i,j}} C_w$ in the path of $C \GLOBTRANS^* C'$.
By Observation~\ref{obs:started-conversations} we get that each such step of type (2) is eventually followed by transitions of type (3) and (4) and since $\TRANSL{C'}$ has only edge states $\SINIT$, this has to happen before $\TRANSL{C'}$ is reached.
Therefore, there exists a $C' \in \POP{Q}$ with $D \GLOBTRANS^{*} \TRANSL{C'}$ and $C \GLOBTRANS^{*} C'$.
\end{proof}
\begin{lemma}\label{lem:input-output-corresponding-translation}
The translation $\TRANSL{\cdot}$ given in Def.~\ref{def:translation} is I/O corresponding for any population protocol $P = (Q, \Sigma, \LOCIN, \LOCOUT, \LOCTRANS)$ and the mediated population protocol $P' = (Q',\Sigma',S',s'_0,\LOCIN',\LOCOUT',\LOCTRANS')$ constructed using Simulation~\ref{sim:PP-to-IOMPP}.
\end{lemma}
\begin{proof}
The input function $\LOCIN'$ of $P'$ makes use of the input function $\LOCIN$ of $P$ by putting an agent in the unlocked state $(U, \LOCIN(\sigma))$ iff the same agent would be in state $\LOCIN(\sigma)$ in $P$.
Also every edge state is initialized with $\SINIT$.
This matches $\TRANSL{\cdot}$ that translates each state $q$ to $(U, q)$ and sets every edge state to $\SINIT$.
Thus, $\TRANSL{\GLOBIN\left(V\right)} = \GLOBIN'\left(V\right)$ and $\TRANSL{\cdot}$ is input corresponding.

The output function $\LOCOUT'$ of $P'$ makes use of the output function $\LOCOUT$ of $P$ by ignoring the locking state of an agent and giving back the output of $\LOCOUT$ for the computation state.
As seen above, $\TRANSL{\cdot}$ translates every agents state $q$ to a tuple with $q$ being the second component \ie its computation state.
For each agent $a$ it holds that $\LOCOUT\left(\left(C\right)_{a}\right) = \LOCOUT'\left(\left(D\right)_{a,a}\right)$.
As $\GLOBOUT$ aggregates the outputs of all agents, which are the same in $P$ and $P'$, $\GLOBOUT\left(C\right) = \GLOBOUT'\left(\TRANSL{C}\right)$ and consequently $\TRANSL{\cdot}$ is output corresponding.
\end{proof}
\begin{lemma}\label{lem:output-stability-preserving-translation}
The translation $\TRANSL{\cdot}$ given in Def.~\ref{def:translation} is output stability preserving for any population protocol $P = (Q, \Sigma, \LOCIN, \LOCOUT, \LOCTRANS)$ and the mediated population protocol $P' = (Q',\Sigma',S',s'_0,\LOCIN',\LOCOUT',\LOCTRANS')$ constructed using Simulation~\ref{sim:PP-to-IOMPP}.
\end{lemma}
\begin{proof}
Let $C$ be a configuration that is output stable in $P$.
Assume towards contradiction $\TRANSL{C}$ is not output stable in $P'$.
From the definition of output stability follows that a configuration $D$ exists with $\GLOBOUT'\left(\TRANSL{C}\right) \not= \GLOBOUT'\left(D\right)$ and $D$ is reachable from $\TRANSL{C}$, \ie there is a path $\TRANSL{C} \GLOBTRANS D_1 \GLOBTRANS D_2 \GLOBTRANS \dots \GLOBTRANS D$.
W.l.o.g assume that $D$ is the first such configuration in this path, \ie $\GLOBOUT'\left(\TRANSL{C}\right) = \GLOBOUT'\left(D_i\right)$ for each such $D_i$.
Since $\GLOBOUT'$ aggregates the values of $\LOCOUT'$ for each agent and because $P'$ is an immediate observation protocol, there has to be a single agent $a$ that has changed its output because of the transition leading to $D$.
From Observation~\ref{obs:output-changing-transitions} we know that this transition has to be of type (1), (2), or (5) and $a$ has to be its observer.

If it is (1) or (2) we can construct a configuration in the same way as $C'$ was constructed in the proof of operational soundness for Lemma~\ref{lem:operationally-corresponding-translation}.
This $C'$ is reachable from $C$ in $P$ and $\TRANSL{C'}$ is reachable from $D$ in $P'$.
Note that on the path from $D$ to $\TRANSL{C'}$ agent $a$ is never an observer of any transition with type (1), (2), or (5) and therefore does not change its output.
Because our translation maintains outputs for each agent $\LOCOUT\left(\left(C\right)_{a}\right) = \LOCOUT'\left(\left(\TRANSL{C}\right)_{a,a}\right) \not= \LOCOUT'\left(\left(D\right)_{a,a}\right) = \LOCOUT'\left(\left(\TRANSL{C'}\right)_{a,a}\right) = \LOCOUT\left(\left(C'\right)_{a}\right)$ holds.
This results in $\GLOBOUT\left(C\right) \not= \GLOBOUT\left(C'\right)$, a contradiction to $C$ being output stable.

If the transition agent $a$ took was of type (5), there has to be a configuration along the path from $\TRANSL{C}$ to $D$ where agent $a$ took a transition of type (1).
Observation~\ref{obs:started-conversations} states that there has to be such a transition in advance and by the definition of $\TRANSL{\cdot}$ this has to be after $\TRANSL{C}$.
But if the transition of type (5) changes the output of $a$, the corresponding type (1) transition must also have changed it, contradicting our assumption of $D$ being the first configuration with an output different from $\GLOBOUT'\left(\TRANSL{C}\right)$.
This is due to a type (5) transition resetting the computation state of an agent back to the state it had before the corresponding type (1) transition.

For the other direction assume towards contradiction $\TRANSL{C}$ is output stable in $P'$ and $C$ is not output stable in $P$.
Then there exists a configuration $\overline{C}$ that is reachable from $C$ with $\GLOBOUT\left(C\right) \not= \GLOBOUT\left(\overline{C}\right)$.
Because $\TRANSL{\cdot}$ is operationally complete by Lemma~\ref{lem:operationally-corresponding-translation} it holds that $\TRANSL{\overline{C}}$ is reachable from $\TRANSL{C}$.
From the output correspondence of $\TRANSL{\cdot}$ in Lemma~\ref{lem:input-output-corresponding-translation} follows that $\GLOBOUT'\left(\TRANSL{C}\right) = \GLOBOUT\left(C\right) \not= \GLOBOUT\left(\overline{C}\right) = \GLOBOUT'\left(\TRANSL{\overline{C}}\right)$.
This is a contradiction to the output stability of $\TRANSL{C}$ in $P'$.
\end{proof}
\begin{theorem}
For any population protocol $P = (Q, \Sigma, \LOCIN, \LOCOUT, \LOCTRANS)$, the immediate observation mediated population protocol $P'$ constructed from $P$ using Simulation~\ref{sim:PP-to-IOMPP} calculates the same semilinear predicate.
\end{theorem}
\begin{proof}
We have shown that $\TRANSL{\cdot}$ from Definition~\ref{def:translation} is operationally corresponding, I/O corresponding and output stability preserving in Lemmas~\ref{lem:operationally-corresponding-translation}, \ref{lem:input-output-corresponding-translation}, and \ref{lem:output-stability-preserving-translation}.
By Lemma~\ref{lem:translation-gives-predicate-equality} the statement directly follows.
\end{proof}
%\begin{proof}
%We have shown that our translation $\TRANSL{\cdot}$ from Definition~\ref{def:translation} is operationally corresponding in Lemma~\ref{lem:operationally-corresponding-translation}, I/O corresponding in Lemma~\ref{lem:input-output-corresponding-translation}, and output stability preserving in Lemma \ref{lem:output-stability-preserving-translation}.
%By Lemma~\ref{lem:translation-gives-predicate-equality} all requirements are met and the statement directly follows.
%\end{proof}
%
\begin{corollary}
Immediate observation mediated population protocols can compute every semilinear predicate and are therefore at least as expressive as population protocols.
\end{corollary}
\subsection{Immediate observation does not restrict the computational power of mediated population protocols}\label{subsec:IOMPP-vs-MPP}
The approach in the previous sections can be used to prove that two-way communication in mediated population protocols does not add to the computational power of the model.
To achieve this we define a simulation of mediated population protocols into the variant with immediate observation communication.
\begin{simulation}\label{sim:MPP-to-IOMPP}
Let $P = (Q, \Sigma, S, s_0, \LOCIN, \LOCOUT, \LOCTRANS)$ be a mediated population protocol. The following immediate observation mediated population protocol $P'$ simulates the protocol $P$ and is given by the tuple $(Q', \Sigma', S', s'_0, \LOCIN', \LOCOUT', \LOCTRANS')$ where
\begin{eqnarray*}
    \Sigma' &:=& \Sigma, \\
    Q' &:=& \{L,U\} \times Q, \\
    S' &:=& \left(\{\SINIT, \SPONR\} \cup \left( Q \times S \right)\right) \times S, \\
    s'_0 &:=& \left(\SINIT, s_0\right), \\
    \LOCIN'(\sigma) &:=& (U,\LOCIN(\sigma)) \text{ for all } \sigma \in \Sigma, \\
    \LOCOUT'(l,q) &:=& \LOCOUT(q) \text{ for all } (l,q) \in Q'.
\end{eqnarray*}
W.l.o.g. we assume that $\{\SINIT, \SPONR\} \cap (Q \cup S) = \emptyset$.
In contrast to Simulation~\ref{sim:PP-to-IOMPP} the edge state has two components.
The first component again signals the current state of the simulated communication and serves as a backup for the condition prior to the communication.
Here we need to save both, computation state and edge state.
The second component represents the actual edge state present in the original protocol.
We specify for each transition $t = (p,r,q,s) \rightarrow_{\LOCTRANS} (p',r',q',s')$ of $\LOCTRANS$ with $p,q,p',q' \in Q$ and $r,s,r',s' \in S$ the following transitions for $\LOCTRANS'$.
\begin{eqnarray*}
	t^{(1)} = ((U,p),~(\SINIT, r),~(U,q),~(\SINIT, s)) &\rightarrow_{\LOCTRANS'}& ((U,p),~(\SINIT, r),~(L,q'),~((q,s),s')) \\
	t^{(2)} = ((L,q'),~((q,s),s'),~(U,p),~(\SINIT, r)) &\rightarrow_{\LOCTRANS'}& ((L,q'),~((q,s),s'),~(L,p'),~(\SPONR, r')) \\
    t^{(3)} = ((L,p'),~(\SPONR, r'),~(L,q'),~((q,s),s')) &\rightarrow_{\LOCTRANS'}& ((L,p'),~(\SPONR, r'),~(U,q'),~(\SINIT,s')) \\
    t^{(4)} = ((x,y),~(\SINIT,s'),~(L,p'),~(\SPONR, r')) &\rightarrow_{\LOCTRANS'}& ((x,y),~(\SINIT,s'),~(U,p'),~(\SINIT, r')) \\
    &&\text{ for every } (x,y) \in Q' \nonumber \\
    t^{(5)} = ((x,y),~(v,w),~(L,q'),~((q,s),s')) &\rightarrow_{\LOCTRANS'}& ((x,y),~(v,w),~(U,q),~(\SINIT,s)) \\
    &&\text{ for every } (x,y) \in Q' \nonumber \\
    &&\text{ and } (v,w) \in S' \setminus \{(\SPONR, r')\} \nonumber
\end{eqnarray*}
\end{simulation}
This simulation follows the same ideas as the Simulation~\ref{sim:PP-to-IOMPP}.
The only difference is the edge state of the original protocol that needs to be taken into account by transitions of the simulation and that needs to be backed up for possible future resets.
We can now give a translation similar to Definition~\ref{def:translation} required for our line of argumentation.
\begin{definition}[Translation of Mediated Configurations]\label{def:translation-MPP-to-IOMPP}
Let $P = (Q, \Sigma, S, s_0, \LOCIN, \LOCOUT, \LOCTRANS)$ be a mediated population protocol and $P' = (Q',\Sigma',S',s'_0,\LOCIN',\LOCOUT',\LOCTRANS')$ a immediate observation mediated population protocol constructed from $P$ using Simulation~\ref{sim:MPP-to-IOMPP}.
By $\TRANSL{\cdot}: \POPM{Q} \to \POPM{Q'}$ we denote the translation of configurations $C \in \POPM{Q}$ in the population protocol into configurations $D \in \POPM{Q'}$ from the mediated population protocol.
This translation is defined as follows:
\[\TRANSL{\begin{pmatrix}
            q_1 & s_{2,1} & \dots & s_{n,1} \\
            s_{1,2} & q_2 & \ddots & \vdots \\
            \vdots & \ddots & \ddots & s_{n,n-1} \\
            s_{1,n} & \dots & s_{n-1,n} & q_n
        \end{pmatrix}
    }
    =
    \begin{pmatrix}
        (U,q_1) & \left(\SINIT, s_{2,1}\right) & \dots & \left(\SINIT, s_{n,1}\right) \\
        \left(\SINIT, s_{1,2}\right) & (U,q_2) & \ddots & \vdots \\
        \vdots & \ddots & \ddots & \left(\SINIT, s_{n,n-1}\right) \\
        \left(\SINIT, s_{1,n}\right) & \dots & \left(\SINIT, s_{n-1,n}\right) & (U,q_n)
    \end{pmatrix}
\]
\end{definition}
With this simulation and translation we conjecture that each mediated population protocol can be simulated by a immediate observation mediated population protocol that shares several attributes, especially computing the same predicates.
\begin{conjecture}\label{conj:MPP-vs-IOMPP}
For any mediated population protocol $P = (Q, \Sigma, S, s_0, \LOCIN, \LOCOUT, \LOCTRANS)$, the immediate observation mediated population protocol $P' = (Q',\Sigma',S',s'_0,\LOCIN',\LOCOUT',\LOCTRANS')$ constructed from $P$ using Simulation~\ref{sim:MPP-to-IOMPP} calculates the same semilinear predicate.
\end{conjecture}
\section{Conclusion and future work}\label{sec:conclusion-and-futer-work}
We have given a proof for the model of immediate observation mediated population protocols to compute all semilinear predicates.
Thus they are as least as powerful in computation as population protocols.
Additionally, we have given arguments why we believe this model is even equivalent to the model of mediated population protocols with two-way communication.
Consequently allowing the initiator of a transition to change its agent and edge states does not contribute to the computational power.
The proof of our Conjecture~\ref{conj:MPP-vs-IOMPP} can hopefully be done in our future research.

Our approach asks for a simulation and a translation which might seem overly complicated for a proof of equal computational power.
But additionally several other attributes, besides the computation of the same predicate, carry over from the one protocol to the other if our Simulation~\ref{sim:PP-to-IOMPP} and translation from Definition~\ref{def:translation} are used.
Consider for example livelock freedom, \ie no configuration is reached that has no successor besides itself.
The simulation can reach a livelock iff the original protocol can reach such a configuration.
This can easily be derived from the operational correspondence in Definition~\ref{def:operational-correspondence}.
In the context of protocols computing some predicate, a livelock is only possible if an output is reached.
Otherwise the requirements for a well-specified protocol are not met.
A livelock can be a desirable state as the computation can clearly be stopped in such a configuration.
If the protocol does something else than computing a predicate, livelocks can be even more important to be reached or avoided, depending on the situation.

Another example is the analysis of a required communication structure.
Whereas some protocols need a full interaction graph to carry out a computation, a path structure would suffice for others to get a correct result.
The interaction graphs supporting a protocol do also support its simulation.

As a last example consider failure resistance \cite{DFG06}.
If a protocol is designed to tolerate a certain number and type of faults, the simulation of this protocol could be capable of a comparable behaviour.
This however depends on the type of failure and the chosen strategy to handle it.
Crash failures, where an agent may leave the population at any time, should be manageable in the simulation with the same mechanisms as the original protocol did.
Message losses could lead to new problems in the simulation like deadlocked communication partners.
Some error handling and error masking strategies could lead to the number of failures tolerable by the simulation being reduced in contrast to the original protocol.

A study on desirable attributes and how they carry over from one protocol to another by our simulation is something we wish to address in the future.

\bibliographystyle{eptcs}
\bibliography{literature}
\end{document}